\documentclass[aps,prd,twocolumn,superscriptaddress]{revtex4-2}

\usepackage{makecell} 
\usepackage{microtype}
\usepackage{graphicx}
\usepackage{dcolumn}
\usepackage{bm}
\usepackage{caption} 
\usepackage{mathtools} 
\usepackage{tensor} 
\usepackage[hidelinks]{hyperref} 
\usepackage{amssymb}
\usepackage{mathtools}
\usepackage{mathrsfs}
\usepackage{amsmath} 
\usepackage{xcolor}
\usepackage{array} 
\usepackage{booktabs}
\usepackage{subcaption}
\usepackage{enumitem}
\usepackage{soul}
\usepackage{amsthm}
\usepackage{stmaryrd}

\usepackage[toc,page]{appendix}

\newtheorem{theorem}{Theorem}

\theoremstyle{remark}

\newcommand{\dd}{\mathrm{d}}

\begin{document}

\title{Generalized Einstein-ModMax-ScalarField theories and new exact solutions}

\author{Leonel Bixano}
    \email{Contact author: leonel.delacruz@cinvestav.mx}
\author{Tonatiuh Matos}%
 \email{Contact author: tonatiuh.matos@cinvestav.mx}
\affiliation{Departamento de F\'{\i}sica, Centro de Investigaci\'on y de Estudios Avanzados del Intituto Politécnico Nacional, Av. Intituto Politécnico Nacional 2508, San Pedro Zacatenco, M\'exico 07360, CDMX.
}%

\date{\today}

\begin{abstract}
We present a generalized Ernst-type framework for stationary, axisymmetric spacetimes in which a scalar field is coupled to the electrodynamic field, \emph{with a particular focus on the ModMax theory}. Our approach relies on the Weyl stationary-axisymmetric ansatz and explicitly allows for a nonzero rotational metric function, $\omega\neq 0$. The resulting setup is broad enough to encompass wide classes of scalar couplings, including dilatonic and phantom-like sectors, and can be tailored to specific models such as Einstein-ModMax, Kaluza–Klein theories, low-energy string-inspired scenarios, entanglement relativity and related generalizations. Within this scheme, we derive two novel families of exact rotating solutions in the sector where the electromagnetic invariants obey $\mathcal F/\mathcal G=\mathrm{constant}$. This regime is particularly significant for ModMax, as it preserves genuinely nonlinear features while still admitting an analytically manageable description.
\end{abstract}

\maketitle

\section{Introduction}
Nonlinear electrodynamics offers a natural framework for investigating how strong-field deviations from Maxwell theory affect both the behavior of electromagnetic fields and the related gravitational dynamics. Among the various nonlinear models that have been introduced, ModMax holds a special status. In four spacetime dimensions, it is the only one-parameter nonlinear modification of Maxwell electrodynamics that retains two of the most stringent symmetries of the linear theory: conformal invariance and continuous electric-magnetic duality invariance \cite{Bandos:2020jsw,Kosyakov:2020wxv}. Owing to this, ModMax has quickly emerged as an especially appealing setting in which to reexamine classical self-gravitating solutions and to assess how much of the exact-solution structure of Einstein-Maxwell theory persists once one moves beyond the linear regime.

From a geometric standpoint, having both conformal symmetry and duality invariance simultaneously is extremely constraining. Conformal invariance implies that, in four dimensions, the electromagnetic sector carries no intrinsic length scale, which sharply distinguishes ModMax from more conventional nonlinear theories such as Born–Infeld. Duality invariance, on the other hand, requires the field equations to remain unchanged under continuous rotations that mix the electric and magnetic components, thus preserving one of the hallmark structural features of source-free Maxwell theory \cite{Bandos:2020jsw,Kuzenko:2024zra,Ayon-Beato:2024vph}. Taken together, these properties render ModMax conceptually remarkable and indicate that it may possess a nontrivial class of exact solutions that is richer than one might initially anticipate for a nonlinear model.

At the same time, the list of known exact self-gravitating solutions in Einstein–ModMax theory is still fairly limited. The first explicitly derived self-gravitating configurations were Reissner–Nordström-type charged black holes and exact gravitational-wave spacetimes, which already displayed the characteristic charge-screening effect produced by the nonlinear parameter \cite{Flores-Alfonso:2020euz}. Shortly thereafter, charged Taub–NUT solutions and their nonlinear generalizations were obtained, along with NUT–wormhole and Taub–Bolt sectors \cite{BallonBordo:2020jtw,Flores-Alfonso:2020nnd}. Accelerating black holes were subsequently constructed in asymptotically AdS backgrounds, yielding some of the earliest instances of accelerated exact solutions in this nonlinear conformal electrodynamics \cite{Barrientos:2022bzm}. More recently, the catalog of explicit solutions has been broadened to encompass Melvin–Bonnor electromagnetic universes, electromagnetized black holes, swirling geometries, exact multi–black-hole configurations, and black diholes generated via generalized Harrison transformations \cite{Barrientos:2024umq,Bokulic:2025usc,Bokulic:2025ucy}.

A recurring characteristic of the exact solutions currently known is that they are typically confined to very restricted sectors—most often static, purely electric, purely magnetic, or obtained through strong algebraic simplifications. In contrast, rotating stationary-axisymmetric configurations, particularly when a scalar field is present, are far less well understood in the context of ModMax.

In this work, we develop a generalized Ernst formalism for stationary, axisymmetric systems that incorporate a scalar field and nonlinear electrodynamics, placing particular emphasis on ModMax and Weyl geometries with a nonvanishing rotational metric function $\omega \neq 0$. We then solve this framework to derive two new families of solutions in the regime characterized by $\mathcal{F}/\mathcal{G}$.

We start by introducing the Lagrangian in the \textit{Einstein frame}, which will be employed throughout this work
\begin{equation}\label{Lagrangiano}
    \mathfrak{L}=\sqrt{-g}\bigg(-\mathcal R +2\epsilon_0 (\nabla \phi)^2 + e^{-2 \alpha_0 \phi } \mathfrak{L}_{MM} \bigg),
\end{equation}
where 
\[
    \mathfrak{L}_{MM}= \mathcal{F} \; \cosh{\gamma} +\sqrt{\mathcal{F}^2+\mathcal{G}^2} \; \sinh{\gamma} ,
\]
\(g = \det(g_{\mu\nu})\) is the determinant of the metric tensor, \(\mathcal R\) denotes the Ricci scalar, \( (\nabla \phi)^2= \nabla_{\mu} \phi \; \nabla^{\mu} \phi \), \(\gamma\) is the deformation parameter, \(\mathcal F = F_{\mu \nu} F^{\mu \nu}\) is the usual Maxwell invariant, \(\mathcal G = F_{\mu \nu} \star F^{\mu \nu}\) stands for its dual invariant, and \(\star F^{\mu \nu}\) denotes the dual Faraday tensor. The scalar field is denoted by \(\phi\), and the parameter \(\epsilon = \pm 1\) serves to distinguish the dilaton scalar field from the phantom scalar field, respectively. The coupling constant \(\alpha_{0} \in \{0, 1, 3, \tfrac{1}{2\sqrt{3}}\}\) identifies the particular theoretical model being studied; for instance, these values correspond, in order, to the Einstein–ModMax (EMM) theory, the low-energy effective superstring (SS) theory, the Kaluza–Klein (KK) theory, and entanglement relativity (ER).\footnote{See \cite{Minazzoli:2025gyw,Minazzoli:2025nbi} .}.
The corresponding field equations are:
\begin{subequations}\label{EcuacionesDeCampoOriginales}
\begin{equation}\label{Eq:Campo1}
    \nabla_\mu \left(  P^{\mu \nu} \right)=0,
\end{equation}
\begin{equation}\label{Eq:Campo2}
    \epsilon_0 \nabla^2 \phi+\frac{\alpha_0}{2}  \left( e^{-2\alpha_0 \phi}\, \mathfrak{L}_{MM} \, \right)=0,
\end{equation}
{\small
\begin{equation}\label{Eq:Campo3}
    \mathcal R_{\mu \nu}=2 \epsilon_0 \nabla_\mu \phi \nabla_\nu \phi  + 2 \, w \, e^{-2\alpha_0 \phi} \left( F_{\mu \sigma} \tensor{F}{_\nu}{^\sigma} -\frac{1}{4} g_{\mu \nu } \mathcal{F} \right),
\end{equation}
}
\end{subequations}
where, in line with the concepts developed in \cite{Ayon-Beato:2024vph}, we made use of the two fundamental symmetries of ModMax electrodynamics: conformal invariance and electric–magnetic duality invariance. Conformal invariance implies that the trace of the ModMax electromagnetic energy–momentum tensor vanishes,
\[
(T^\mu{}_\mu)_{MM}=0,
\]
then, the trace of equation \eqref{Eq:Campo3} is given by $\mathcal R=2\epsilon_0 (\nabla \, \phi)^2 $, while duality invariance guarantees that its traceless component remains unaffected by duality rotations. Moreover, in four spacetime dimensions one can employ the algebraic identity
\[
F_{\mu\rho}\,{}^\star F_{\nu}{}^{\rho}
-\frac14 g_{\mu\nu}\,\mathcal G=0.
\]
In complete analogy, we have introduced the following variables to simplify the ensuing computations
\begin{subequations}
\begin{align}
    &\kappa^2=e^{-2\alpha_0 \, \phi }, \\
    &\mathtt{X} = \kappa^2 \, \mathcal{F}, \\
    &\mathtt{Y} = \kappa^2 \, \mathcal{G}, \\
    &\Delta =\sqrt{\mathtt{X}^2+\mathtt{Y}^2},\\
    &\tan{\Theta}=\frac{\mathtt{Y}}{\mathtt{X}},
\end{align}
\begin{align}
    w(\Theta)&=\cosh{\gamma} +\cos{\Theta} \, \sinh{\gamma} ,
\\
v(\Theta)&=\sin{\Theta} \, \sinh{\gamma} ,
\end{align}
    \begin{equation}
        P^{\mu \nu}=\kappa^2 \Big[ w\, F^{\mu \nu} + v \,  \star F^{\mu \nu}\Big].
    \end{equation}
\end{subequations}
Using this representation the ModMax part of the lagrangian takes the form 
\[
    e^{-2 \alpha_0 \phi } \mathfrak{L}_{MM}=w \, \mathtt{X}+v \, \mathtt{Y}, 
\]
therefore
\[
    \frac{\partial}{\partial \mathcal{F}}\mathfrak{L}_{MM}=w, \qquad \frac{\partial}{\partial \mathcal{G}}\mathfrak{L}_{MM}=v.
\]

It is straightforward to observe that turning off the deformation parameter by setting \(\gamma = 0\) leads to \(w = 1\) and \(v = 0\). Consequently, we obtain \(\mathfrak{L}_{MM} = \mathtt{X}\), \(P^{\mu \nu}=\kappa^2 F^{\mu \nu}\), which reproduces exactly the formulation presented in \cite{Bixano:2026xum}.
\paragraph{Convention for the ModMax sector.}
Within the usual ModMax framework, the electromagnetic Lagrangian takes the form
\[
\tilde{\mathfrak{L}}_{MM}=\mathcal{F}\cosh\gamma-\sqrt{\mathcal{F}^2+\mathcal{G}^2}\,\sinh\gamma .
\]
In this work, we employ an equivalent parametrization of the complete Einstein–scalar–electromagnetic action, tailored to our sign conventions.The translation from our convention to the usual ModMax formulation is achieved by making the replacement
\[
\gamma \mapsto -\gamma .
\]
Accordingly, all formulas stay the same, as long as one simultaneously substitutes the auxiliary coefficients by
\[
v\mapsto -v,
\qquad
w\mapsto \tilde w=\cosh\gamma-\cos\Theta\,\sinh\gamma .
\]
No additional changes are necessary, because all of the reduced equations involve the ModMax sector solely via these coefficients and the derived quantities they generate.

We consider a stationary, axisymmetric space-time, meaning that the geometry possesses two Killing vectors, $Killing = \{\partial_t, \partial_\varphi\}$. Incorporating these symmetries, we can employ the following Weyl ansatz for the metric:
\begin{equation}\label{ds Cilindricas}
     ds^2=-f\left ( dt-\omega d \varphi \right )^2+f^{-1} \left ( e^{2k} (d\rho ^2 + dz^2) +\rho^2 d\varphi^2 \right ), 
\end{equation}
In this setup, the metric functions $\{f, \omega, \kappa\}$ are taken to depend on the coordinates $(\rho,z)$. We also employ an axisymmetric four-potential ansatz of the form $A_{\mu} = \big[ A_t(\rho,z),\,0,\,0,\,A_{\varphi }(\rho,z) \big]$, so that the scalar field is likewise a function of these coordinates, $\phi(\rho,z)$.

The spheroidal projection is presented in Appendix \ref{Apendix:Coordinates protection}.
\section{Potentials method}
It is important to emphasize that the construction of the formalism developed in this work is based on the article \cite{Bixano:2026xum}, which itself is based on the earlier studies \cite{Matos:2000ai,Matos:2000za}, where Ernst-type potentials were rigorously generalized to an \emph{Einstein-Maxwell-Scalar field} setting. In the present paper, we further broaden this framework to encompass an \emph{Einstein-ModMax-Scalar field} formalism.

Consider two operators that are orthogonal to each other (that is, they satisfy $D\tilde{D} = 0$):
\begin{equation}\label{OperadoresDDguiño}
    D\equiv \begin{bmatrix}
        \partial_\rho \\
        \partial_z
    \end{bmatrix}, \qquad 
    \Tilde{D}=
    \begin{bmatrix}
    \begin{array}{c}
    \partial_z \\
    -\partial_\rho
    \end{array}
    \end{bmatrix},
\end{equation}
which has the following properties \( \tilde{\tilde{D}} =-D, \quad  \Tilde{D}A\Tilde{D}B=DADB, \quad \Tilde{D}ADB=-\Tilde{D}BDA, \quad \tilde{D}ADA=0, \quad \Tilde{D}\Tilde{D}=DD=\partial_\rho ^2 +\partial_z^2,\) and we shall adopt the notation $DDf=D^2f$, and $Df Df=Df^2$.
Taking into account the operators (\ref{OperadoresDDguiño}), the definition $\kappa^2 = e^{-2\alpha_0 \phi}$, and the metric (\ref{ds Cilindricas}), we can rewrite (\ref{EcuacionesDeCampoOriginales}) in the following form:
{\small
\begin{subequations}\label{EcuacionesDeCampoV1}
\begin{center}
    \text{Klein-Gordon Equation:}
    \begin{multline}
        D^2\kappa +D\kappa \left( \frac{D\rho}{\rho} - \frac{D\kappa}{\kappa} \right)+ v \; \frac{\alpha_0^2 \kappa^3}{\rho \epsilon_0 }2\tilde{D} A_t DA_\varphi \\
        -w\;\frac{f \kappa^3 \alpha_0^2 }{\rho^2 \epsilon_0 }  \left[ L^2(\frac{\omega}{L} DA_t +DA_\varphi)^2-\frac{\rho^2}{f^2} DA_t^2 \right] =0 \label{eq:KleinGordonV1} ,
    \end{multline}
    \text{Maxwell Equations:}
    \begin{align}
        &D \left[ w\;\frac{2 \kappa^2 f}{\rho}  (\omega DA_t +DA_\varphi)-2v\; \kappa^2\tilde{D}A_t \right]=0, \label{Eq:Maxwell1V1} \\ 
        &D \bigg[ 2 w\;  \kappa^2 \left( \frac{f \omega}{\rho}  (\omega DA_t +DA_\varphi) -\frac{\rho}{f} DA_t \right) \nonumber \\
        &\qquad \qquad +2v\;\kappa^2 \tilde{D}A_\varphi \bigg]=0, \label{Eq:Maxwell2V1}
    \end{align}  
    \text{Einstein Equations:}
    \begin{multline}
        D^2f + Df\left[ \frac{D\rho}{\rho} - \frac{Df}{f} \right]+\frac{f^3}{\rho^2}D\omega^2 \\
        -w\; \frac{2  \kappa^2 f^2}{\rho^2}\left[ (\omega DA_t +DA_\varphi)^2 +\frac{\rho^2}{f^2} DA_t^2 \right]=0, \label{Eq:Einstein1V1}
    \end{multline}
    \begin{equation}
        D^2\omega - D\omega \left[ \frac{D\rho}{\rho} - \frac{2 Df}{f} \right]+w\; \frac{4  \kappa^2}{f}   (\omega DA_t +DA_\varphi) DA_t=0. \label{Eq:Einstein2V1}
    \end{equation}
\end{center}
\end{subequations}
}

Introducing the notation $[MM\text{-}E]^{\nu} = \nabla_\mu \left( P^{\mu \nu} \right)$, equation (\ref{Eq:Maxwell1V1}) is obtained from the $\varphi$-component $[MM\text{-}E]^{\varphi} = 0$, while equation (\ref{Eq:Maxwell2V1}) arises from the time component $[MM\text{-}E]^{t} = 0$. Next, defining $[E\text{-}E]_{\mu \nu} = R_{\mu \nu} - 2 \epsilon_0 \nabla_\mu \phi \nabla_\nu \phi - 2w\, \kappa^2 \left( F_{\mu \sigma} \tensor{F}{_\nu}{^\sigma} - \frac{1}{4} g_{\mu \nu} F^2 \right)$, equation (\ref{Eq:Einstein1V1}) follows from the $tt$-component $[E\text{-}E]_{tt} = 0$, whereas equation (\ref{Eq:Einstein2V1}) is derived from the linear combination $[E\text{-}E]_{\varphi t} + \omega \,[E\text{-}E]_{tt} = 0$.
\subsection{Definition of the potentials}
By making use of equations (\ref{Eq:Maxwell1V1}) and (\ref{Eq:Einstein2V1}), two correspondingly important potentials can be introduced
\begin{subequations}\label{DefinicionPotencialesGeneral}
\begin{align}
    &\Tilde{D}\chi=w\, \tilde{D}\chi_{MB}-v \; \kappa^2 \, \tilde{D} \psi, \label{TDchi}\\
    &\Tilde{D} \epsilon = \frac{f^2}{\rho} D\omega + \psi \Tilde{D}\chi \label{TDepsilon},
\end{align}
\end{subequations}
where
\[
    \Tilde{D}\chi_{MB}=\frac{2f\kappa^2}{\rho}  (\omega DA_t +DA_\varphi),
\]
which is associated with the $\chi$-potential introduced by Matos-Bixano in \cite{Bixano:2026xum}. The presence of these two potentials is equivalent to satisfying the first ModMax equation and the second Einstein equation, with $\psi = 2A_t$.
Therefore, we introduce the \textbf{potentials} $Y^A$, defined as
\begin{equation}\label{Potenciales}
    (Y^A)^T=\Big[f,\epsilon,\psi,\chi,\kappa \Big],
\end{equation}
which represent, in order, the gravitational, rotational, electric, magnetic, and scalar potentials.
By inserting the definitions (\ref{DefinicionPotencialesGeneral}) into the field equations (\ref{EcuacionesDeCampoV1}), the latter simplify to
{\small
\begin{subequations}\label{EcuacionesDeCampoV2}
\begin{center}
    \text{Klein-Gordon equation}
    \begin{equation}
        D(\rho D\kappa ) - \frac{\rho }{\kappa} D\kappa^2+\frac{\rho \kappa^3 \alpha_0^2}{4 f \epsilon_0}  \left( \frac{v^2+w^2}{w}D\psi^2-\frac{D\chi^2}{w\kappa^4 }  \right) =0 \label{eq:KleinGordonV2} ,
    \end{equation}
    \text{Maxwell equations}
    \begin{multline}\label{Eq:Maxwell1V2}
        D \left[ \frac{\rho}{f} \left( \frac{v}{w} D\chi + \frac{v^2+w^2}{w}\kappa^2 D\psi \right) \right] \\
        - \frac{\rho}{f^2} (D\epsilon-\psi D \chi)D\chi=0 ,
    \end{multline}
    \begin{multline}\label{Eq:Maxwell2V2}
        D \left[ \frac{\rho}{f} \left( \frac{1}{w \kappa^2} D\chi + \frac{v}{w} D\psi \right) \right] \\
        + \frac{\rho}{f^2} (D\epsilon-\psi D \chi)D\psi=0 , 
    \end{multline}  
    \text{Einstein equations}
    \begin{multline}\label{Eq:Einstein1V2}
        D (\rho Df)+\frac{\rho}{f}\left( (D\epsilon-\psi D\chi )^2 -Df^2 \right)  \\
        -\frac{\rho }{2} \left( \frac{v^2+w^2}{w}\kappa^2 D\psi^2+\frac{D\chi^2}{w\kappa^2 }+2\frac{v}{w}D\psi D\chi  \right) = 0 ,
    \end{multline}
    \begin{equation}\label{Eq:Einstein2V2}
        D (\rho (D\epsilon-\psi D\chi)) - \frac{2 \rho }{f} (D\epsilon-\psi D\chi) Df =0. 
    \end{equation}  
\end{center}
\end{subequations}
}
Equations (\ref{Eq:Maxwell2V2}) and (\ref{Eq:Einstein2V2}) follow from the analyticity of the metric functions $A_t$ and $\omega$, i.e. $\Tilde{D} D A_\varphi=0$ and $\Tilde{D} D \omega=0$, respectively.

The field equations in (\ref{EcuacionesDeCampoV2}) can be rederived from the Euler–Lagrange expression:
\begin{equation}\label{Eq:EulerLagrange}
    D\left(  \frac{\partial \mathfrak{L}}{\partial(DY^A)} \right)
    - \frac{\partial \mathfrak{L}}{\partial Y^A} = 0,
\end{equation}
by employing the following Lagrangian density:
\begin{multline}\label{LagrangianoTesis3}
    \mathfrak{L}=\frac{\rho}{2f^2}\left( Df^2 + (D\epsilon-\psi D\chi)^2\right) + \frac{2 \rho \epsilon_0 }{\alpha_0^2 \kappa^2}D\kappa^2 \\
    -\frac{\rho}{2f} \left( \frac{v^2+w^2}{w}\kappa^2 D\psi^2+\frac{D\chi^2}{w\kappa^2 }+2\frac{v}{w}D\psi D\chi  \right).
\end{multline}
\subsection{Space of potentials}
From (\ref{LagrangianoTesis3}), the line element of the associated potential space is:
\begin{multline}\label{ds Transformado}
    ds^2=\frac{1}{2f^2}\left( \dd f^2 + (\dd \epsilon-\psi \dd \chi)^2\right) + \frac{2 \epsilon_0 }{\alpha_0^2 \kappa^2}\dd \kappa^2 \\
    -\frac{1}{2f}  \left( \frac{v^2+w^2}{w}\kappa^2 \dd \psi^2+\frac{\dd \chi^2}{w\kappa^2 }+2\frac{v}{w}\dd \psi \dd \chi  \right),
\end{multline}
To proceed with the method introduced in \cite{Bixano:2026xum}, we begin by employing the orthonormal coframe given by
\begin{equation}
    \begin{aligned}
        &\theta^0 = \frac{\dd f}{\sqrt2\,f}, \qquad \theta^1 = \frac{\dd \epsilon - \psi\,\dd \chi }{\sqrt2\,f}, \\
        &\theta^2 = \frac{\kappa\sqrt w}{\sqrt{2f}}\,\dd \psi, \qquad \theta^4 = \frac{\sqrt2}{\alpha_0}\frac{\dd \kappa}{\kappa}, \\
        &\theta^3 = \frac{1}{\kappa\sqrt{2fw}}\Big(\dd \chi+\kappa^2 v\, \dd \psi\Big),
    \end{aligned}
\end{equation}
thus, the line element is given by 
\(
\dd s^2 = (\theta^0)^2 + (\theta^1)^2 - (\theta^2)^2 - (\theta^3)^2 + \epsilon_0(\theta^4)^2
\), and the orthonormal vector basis \(\theta^a(e_b)=\delta^a_b\) is
{\small
\begin{equation}\label{eq:Base-Vectorial-Ortonormal}
    \begin{aligned}
        & e_0=\sqrt2\,f\,\partial_f,
 \quad
 e_1=\sqrt2\,f\,\partial_\epsilon, \quad e_4=\frac{\alpha_0\kappa}{\sqrt2}\,\partial_\kappa,\\
 &  e_3=\sqrt{2fw}\,\kappa\,(\partial_{\chi}+\psi\,\partial_\epsilon), \\
 &e_2=\frac{\sqrt{2f}}{\kappa\sqrt w}
 \Big[\partial_\psi-\kappa^2 v\,(\partial_{\chi }+\psi\,\partial_\epsilon)\Big].
    \end{aligned}
\end{equation}
}

As for the Killing vectors of the potential space \eqref{ds Transformado}, they are
{\small
\begin{equation}\label{eq:killing_vectors}
    \begin{aligned}
        &K_\epsilon=\partial_\epsilon,
        \qquad
        K_\chi=\partial_{\chi},
        \qquad
        K_\psi=\partial_\psi+\chi\,\partial_\epsilon,
        \\
        &K_g=2f\,\partial_f+2\epsilon\,\partial_\epsilon+\psi\,\partial_\psi+\chi\,\partial_{\chi},
        \\
        &K_\kappa=\kappa\,\partial_\kappa-\psi\,\partial_\psi+\chi\,\partial_{\chi}.
    \end{aligned}
\end{equation}
}
This follows from the fact that, due to the structure of the invariants, $\mathcal{F},\mathcal{G}$ transform homogeneously under these generators, where the invariants expressed in terms of the potentials are
{\small
\[
    \mathcal{F}=\frac{e^{-2k}}{2 \kappa^2}(\frac{\tilde{D}\chi_{MB}^2}{\kappa^2}-\kappa^2 D\psi^2),\qquad \mathcal{G}=-\frac{e^{-2k}}{2 \kappa^2}(D\psi \; D\chi_{MB}),
\]
}
giving rise to \(v=v(\kappa,k,D\psi,\tilde{D}\psi,D\chi,\tilde{D}\chi;\gamma)\), and \(w=w(\kappa,k,D\psi,\tilde{D}\psi,D\chi,\tilde{D}\chi;\gamma)\). It should be emphasized that the metric function is denoted by $k(\rho,z)$, while the scalar potential is represented by $\kappa(\rho,z)$.
\subsection{New ModMax Anzat}
To determine the associated variables $A,B,C$, we employ a Newman–Penrose-type null coframe, which serves as the basis for constructing these quantities, in line with the framework presented in \cite{Bixano:2026xum}. This coframe is given by:
{\small
\begin{equation}\label{eq:Coframe-NP-like}
    \begin{aligned}
        & \ell =\frac{\theta^0+\theta^1}{\sqrt{2}}, \quad n=\frac{\theta^0-\theta^1}{\sqrt{2}}, \quad m=\frac{\theta^2 + i \theta^3}{\sqrt{2}}, \quad s=\theta^4,
    \end{aligned}
\end{equation}
}
the metric can therefore be written as \(\mathrm ds^2=2\,\ell\,n-2\,m\,\bar m+\epsilon_0\,s^2\).

In this situation, the vector basis of the Newman–Penrose-type coframe is given by:
{\small
\begin{equation}\label{eq:Base-Vectorial-Nula}
    \begin{aligned}
        & e_\ell=f\,(\partial_f+\partial_\epsilon),
 \qquad
 e_n=f\,(\partial_f-\partial_\epsilon), \qquad e_s=\frac{\alpha_0\kappa}{\sqrt2}\,\partial_\kappa,\\
 & e_m=\frac{\sqrt f}{\kappa\sqrt w}
 \Big[\partial_\psi-\kappa^2\,(v+i\,w)\,(\partial_{\chi}+\psi\,\partial_\epsilon)\Big], \\
 & e_{\bar m}=\frac{\sqrt f}{\kappa\sqrt w}
 \Big[\partial_\psi-\kappa^2\,(v-i\,w)\,(\partial_{\chi}+\psi\,\partial_\epsilon)\Big].
    \end{aligned}
\end{equation}
}

The complete formulas for the connection 1-forms, the curvature 2-forms, and the corresponding Ricci invariant are provided in Appendix \ref{Appendix:Structure of the space of potentials}. What is important here is that, in the \emph{frozen ratio case} where \(\mathcal F / \mathcal G\) is constant, \(\Theta\) is likewise constant. This, in turn, implies that \(v = v_0\) and \(w = w_0\) are also constants. Under these conditions, using the definition \eqref{eq:Def_Lambda_Nu} into \eqref{eq:2Forma-Curvatura}, we can obtain the invariant scalars, which takes the form
\begin{equation}\label{eq:InvariantesCongelado}
\begin{aligned}
    \mathcal R
    &=
    -12
    -\frac{\alpha_0^2}{\epsilon_0}
    \left(
        1+\eta_0^2
    \right),
    \\
    \mathcal R_{IJ}\mathcal R^{IJ}
    &=
    36
    +4
    \frac{\alpha_0^4}{\epsilon_0^2}
    \left[
        1
        +4\eta_0^2
        +3\eta_0^4
    \right],
    \\
    \mathcal R_{IJKL}\mathcal R^{IJKL}
    &=
    48
    +4
    \frac{\alpha_0^4}{\epsilon_0^2}
    \left[
        3
        +14\eta_0^2
        +11\eta_0^4
    \right],
    \\
    \mathcal C_{IJKL}\mathcal C^{IJKL}
    &=
    24
    +
    \frac{2\alpha_0^4}{3\epsilon_0^2}
    \left[
        11
        +54\eta_0^2
        +43\eta_0^4
    \right]
    \\
    &\quad + 8\frac{\alpha_0^2}{\epsilon_0}
    \left[ 
        1+\eta_0^2
    \right].
\end{aligned}
\end{equation}
Therefore, this potential space is maximally symmetric and conformally flat.

Now, by applying the equalities \(A=\frac{1}{2}((1-i)\ell +(1+i)n)\), \( B=-\overline{m}\), and \( C=-\frac{\alpha_0}{\sqrt{2}} s \), we define the variables as follows:
\begin{subequations}\label{Variables ABC}
    \begin{equation}\label{VariableA}
        A=\frac{1}{2f}\left[ Df -i (D\epsilon - \psi D\chi) \right],
    \end{equation}
    \begin{equation}\label{VariableB}
        B=-\frac{1}{2\sqrt{f}}\left[  \sqrt{w}\, \kappa D\psi -i \frac{1}{\sqrt{w}\,  \kappa} \Big( D\chi +v\, \kappa^2 D\psi \Big) \right],
    \end{equation}
    \begin{equation}\label{VariableC}
        C=-\frac{D\kappa}{\kappa}.
    \end{equation}
\end{subequations}
By substituting the definition \eqref{Variables ABC} into \eqref{EcuacionesDeCampoV2} and taking into account the definition \eqref{eq:Def_Lambda_Nu}, we obtain the following representation of the field equations:
\begin{subequations}\label{EcuacionesDeCampo-ABC}
 \begin{center}
     \text{Field equations in terms of $\{ A,B,C\}$:}
 \end{center}
    \begin{equation}
        \frac{1}{\rho} D(\rho A)=B\overline{B} +A(A-\overline{A}),\label{EcuacionDeCampoA}
    \end{equation}
    \begin{align}
        \frac{1}{\rho} D(\rho B)&=-\frac{B}{2}(A-3\overline{A})+C\overline{B} \nonumber
        \\
        &\quad -\frac{\lambda}{2}\overline{B}+\frac{i}{2}(\overline{B}-B)\left[  \nu - 2 \frac{v}{w} C\right],\label{EcuacionDeCampoB}
    \end{align}
    \begin{equation}
        \frac{1}{\rho} D(\rho C)=\frac{\alpha_0^2}{2\epsilon_0}\left(B^2+\overline{B}^2 +i\frac{v}{w} (B^2-\overline{B}^2)\right),\label{EcuacionDeCampoC}
    \end{equation}
\end{subequations}
It is worth emphasizing that, in equations \eqref{EcuacionesDeCampo-ABC}, the additional contributions introduced by the ModMax extension are precisely those involving $\frac{v}{w}, \lambda,$ and $\nu$. If we freeze the system, namely by requiring that $\mathcal F / \mathcal G$ remain constant, then $\lambda = 0$ and $\nu = 0$, which leads to a substantial simplification of the equations. Likewise, by choosing $v = 0$, we once again recover the field equations in terms of $A, B, C$ presented in \cite{Bixano:2026xum}.

It is worth noting that equation \eqref{EcuacionDeCampoA} corresponds to the \emph{generalized Ernst potential equation} including a coupled scalar field within the ModMax framework, and in any theory. The phenomenological properties of ModMax are encoded in \eqref{EcuacionDeCampoB} and in the potential $\chi$. Consequently, this formulation provides a compact representation of all the relevant information and makes the underlying structure apparent. For a comprehensive discussion of this discussion, see Ref.~\cite{Bixano:2026xum}.
\section{Subspace projection}
In the frozen case where $\mathcal F/\mathcal G$ is constant, we impose
\[
    v = v_0, \quad w = w_0, \quad \lambda = \nu = 0,
\]
with $v_0, w_0 \in \mathbb{R}$, so that we can project the equations \eqref{EcuacionesDeCampo-ABC} onto a general subspace ($\xi=\mathfrak{s}+i \mathfrak{t},\overline{\xi}=\mathfrak{s}-i \mathfrak{t}$), given by
\begin{equation}\label{MetricaEspacioGeneral}
    ds^2=\frac{2(d\mathfrak{s}^2+d\mathfrak{t}^2)}{(1-\sigma [\mathfrak{s}^2+\mathfrak{t}^2] )^2}=\frac{d\xi d \overline{\xi}}{(1-\sigma \xi \overline{\xi})}, 
\end{equation}
where it is required to satisfy the Laplace equation in curved spaces, i.e.
\begin{equation}\label{Eq:LaplaceGeneraLambda}
    D(\rho D\mathfrak{s}^p)+\rho \tensor{\Gamma}{^p}{_{ij}} D\mathfrak{s}^i D\mathfrak{s}^j =0.
\end{equation}
where $\mathfrak{s}^p$ belongs to the set $\{ \mathfrak{s},\mathfrak{t}\}$, or, stated in another way
\begin{subequations}\label{Eq:LaplaceXi}
\begin{align}
    \frac{1}{\rho}D(\rho D \xi) &= \frac{-2\sigma}{1-\sigma \xi \, \overline{\xi}} \, \overline{\xi} \, (D \xi )^2 , \\
    \frac{1}{\rho}D(\rho D \, \overline{\xi}) &= \frac{-2\sigma}{1-\sigma \xi \, \overline{\xi}} \, \xi \, (D \overline{\xi} )^2 .
\end{align}
\end{subequations}
For an in-depth discussion of this concept, refer to \cite{Bixano:2025bio,Matos:2000za,Matos:2000ai}.
Therefore, in the frozen regime where $\mathcal F/ \mathcal G$ remains constant, the field equations \eqref{EcuacionesDeCampo-ABC}, when projected onto the subspace \eqref{MetricaEspacioGeneral}, take the form
\begin{widetext}
\begin{subequations}\label{N_EcuacionesDeCampo-abc}
    \begin{center}
    \text{Expressions of the extended field equations in terms of $\{ \mathtt{x},\mathtt{y},\mathtt{z}\}$}
    \end{center}
    \begin{align}
        \mathtt{x}_1(\mathtt{x}_1-\overline{\mathtt{x}}_2)+\mathtt{y}_1\overline{\mathtt{y}}_2&=\mathtt{x}_1\bigg(\ln \Big\{ (1-\sigma\,\xi\,\overline{\xi})^{2} \mathtt{x}_1 \Big\}\bigg)_{,\xi}, \label{N_EC-a1} \\
        \mathtt{x}_2(\mathtt{x}_2-\overline{\mathtt{x}}_1)+\mathtt{y}_2\overline{\mathtt{y}}_1&=\mathtt{x}_2\bigg(\ln \Big\{ (1-\sigma\,\xi\,\overline{\xi})^{2} \mathtt{x}_2 \Big\}\bigg)_{,\overline{\xi}}, \label{N_EC-a2} \\
        \mathtt{x}_1(\mathtt{x}_2-\overline{\mathtt{x}}_1)+\mathtt{x}_2(\mathtt{x}_1-\overline{\mathtt{x}}_2)+(\mathtt{y}_1\overline{\mathtt{y}}_1+\mathtt{y}_2\overline{\mathtt{y}}_2)&=(\mathtt{x}_1)_{,\overline{\xi}}+(\mathtt{x}_2)_{,\xi}, \label{N_EC-a3}
    \end{align}
    \begin{align}
        \frac{\mathtt{y}_1}{2}(3\overline{\mathtt{x}}_2-\mathtt{x}_1)+\mathtt{z}_1\overline{\mathtt{y}}_2-i \eta_0 \, \mathtt{z}_1 \, (\overline{\mathtt{y}}_2-\mathtt{y}_1)&=\mathtt{y}_1\bigg(\ln \Big\{ (1-\sigma\,\xi\,\overline{\xi})^{2} \mathtt{y}_1 \Big\}\bigg)_{,\xi}, \label{N_EC-b1} \\
        \frac{\mathtt{y}_2}{2}(3\overline{\mathtt{x}}_1-\mathtt{x}_2)+\mathtt{z}_2\overline{\mathtt{y}}_1-i \eta_0 \, \mathtt{z}_2 \, (\overline{\mathtt{y}}_1-\mathtt{y}_2)
        &=\mathtt{y}_2\bigg(\ln \Big\{ (1-\sigma\,\xi\,\overline{\xi})^{2} \mathtt{y}_2 \Big\}\bigg)_{,\overline{\xi}}, \label{N_EC-b2} \\
        \frac{\mathtt{y}_1}{2}(3\overline{\mathtt{x}}_1-\mathtt{x}_2)+\frac{\mathtt{y}_2}{2}(3\overline{\mathtt{x}}_2-\mathtt{x}_1)+\mathtt{z}_1\overline{\mathtt{y}}_1+\mathtt{z}_2\overline{\mathtt{y}}_2&-i \eta_0 \,\Big[ \mathtt{z}_1 \, (\overline{\mathtt{y}}_1-\mathtt{y}_2) +\mathtt{z}_2 \, (\overline{\mathtt{y}}_2-\mathtt{y}_1) \Big]\nonumber
        \\
        &=(\mathtt{y}_1)_{,\overline{\xi}}+(\mathtt{y}_2)_{,\xi}, \label{N_EC-b3}
    \end{align}
    \begin{align}
        \frac{\alpha_0^2}{2\epsilon_0}\left( \tensor{\mathtt{y}}{_1}{^2}+\tensor{\overline{\mathtt{y}}}{_2}{^2} +i \eta_0 \, (\tensor{\mathtt{y}}{_1}{^2}-\tensor{\overline{\mathtt{y}}}{_2}{^2})\right)=\mathtt{z}_1\bigg(\ln \Big\{ (1-\sigma\,\xi\,\overline{\xi})^{2} \mathtt{z}_1 \Big\}\bigg)_{,\xi},\label{N_EC-c1} 
        \\
        \frac{\alpha_0^2}{2\epsilon_0}\left( \tensor{\mathtt{y}}{_2}{^2}+\tensor{\overline{\mathtt{y}}}{_1}{^2} +i \eta_0 \, (\tensor{\mathtt{y}}{_2}{^2}-\tensor{\overline{\mathtt{y}}}{_1}{^2})\right)=\mathtt{z}_2\bigg(\ln \Big\{ (1-\sigma\,\xi\,\overline{\xi})^{2} \mathtt{z}_2 \Big\}\bigg)_{,\overline{\xi}},\label{N_EC-c2} 
        \\
        \frac{\alpha_0^2}{\epsilon_0}\left(\mathtt{y}_1 \mathtt{y}_2+\overline{\mathtt{y}}_1\overline{\mathtt{y}}_2 +i \eta_0 \, (\tensor{\mathtt{y}}{_1}\tensor{\mathtt{y}}{_2}-\overline{\mathtt{y}}{_1}\overline{\mathtt{y}}{_2})  \right)=(\mathtt{z}_1)_{,\overline{\xi}}+(\mathtt{z}_2)_{,\xi},\label{N_EC-c3} 
    \end{align}
\end{subequations}
\end{widetext}
where we have made use of the fact that
\begin{subequations}\label{AnzatABC}
\begin{align}
    A&=\mathtt{x}_1(\xi, \overline{\xi}\,) D\xi+\mathtt{x}_2(\xi, \overline{\xi}\,) D\, \overline{\xi},\\
    B&=\mathtt{y}_1(\xi, \overline{\xi}\,) D\xi+\mathtt{y}_2(\xi, \overline{\xi}\,) D\, \overline{\xi},\\
    C&=\mathtt{z}_1(\xi, \overline{\xi}\,) D\xi+\mathtt{z}_2(\xi, \overline{\xi}\,) D\, \overline{\xi},
\end{align}
\end{subequations}
and the notation $\eta_0=v_0 /w_0$. By choosing $\sigma = 0$, and thus focusing on an abstract flat subspace, the system simplifies to the equations derived in \cite{Matos:2000za,Matos:2010pcd}, now formulated for the \emph{ModMax} extension.

The relations between the potentials and the variables $\mathtt{x,y,z}$ transform to:
\begin{subequations}\label{N_RelacionesPotenciales-abc}
    \begin{center}
    \text{Relations between the potentials and $\{ \mathtt{x,y,z}\}$}
    \end{center}
    \text{\textbf{$\mathtt{x}$:}}
    \begin{align}
        \mathtt{x}_1=\frac{1}{2f}\big[ f_{,\xi} -i (\epsilon_{,\xi} - \psi \chi_{,\xi}) \big],\label{N_a1} 
        \\
        \mathtt{x}_2=\frac{1}{2f}\big[ f_{,\overline{\xi}} -i (\epsilon_{,\overline{\xi}} - \psi \chi_{,\overline{\xi}}) \big],\label{N_a2} 
    \end{align}
    \text{\textbf{$\mathtt{y}$:}}
    \begin{align}
        \mathtt{y}_1=-\frac{1}{2\sqrt{f}}\left[  \, \kappa \,\sqrt{w_0}\, \psi_{,\xi} -\frac{i}{\sqrt{w_0}\,\kappa} (\chi_{,\xi}+v_0 \, \kappa^2 \, \psi_{\xi}) \right],\label{N_b1} 
        \\
        \mathtt{y}_2=-\frac{1}{2\sqrt{f}}\left[  \, \kappa \,\sqrt{w_0}\, \psi_{,\overline \xi} -\frac{i}{\sqrt{w_0}\,\kappa} (\chi_{,\overline \xi}+v_0 \, \kappa^2 \, \psi_{\overline \xi}) \right],\label{N_b2} 
    \end{align}
    \text{\textbf{$\mathtt{z}$:}}
    \begin{align}
        \mathtt{z}_1=-\frac{1}{\kappa} \kappa_{,\xi},\label{N_c1} \\
        \mathtt{z}_2=-\frac{1}{\kappa} \kappa_{,\overline{\xi}}\label{N_c2}.
    \end{align}
\end{subequations}
\section{Kerr–Newman solution in the frozen ModMax framework}\label{Kerr–Newman solution in the frozen ModMax framework}
It is worth stressing that, by setting $\kappa=1$ and $\alpha_0=0$ in \eqref{N_EcuacionesDeCampo-abc}, that is, by restricting ourselves to the frozen ModMax theory with no coupled scalar field, we recover exactly the same field equations as those reported in \cite{Bixano:2026xum}. This happens because the \emph{ModMax-only contributions} correspond to the terms $i \eta \, \mathtt{z}_i(\overline{\mathtt{y}}_j-\mathtt{y}_p)$ in the $\mathtt{Y}$-block and $i \eta \, (\mathtt{y}_i \mathtt{y}_j-\overline{\mathtt{y}}_p \overline{\mathtt{y}}_q)$ in the $\mathtt{Z}$-block. Put differently, when the scalar field is turned off, the ModMax terms in the $\mathtt{XYZ}$-equations vanish and remain only in the definition of the $\chi$-potential, or, in the alternative case, we choose $\overline{\mathtt{y}}_i=\mathtt{y}_j$ with $i \neq j$.
\subsection{Kerr Newman solution in Maxwell framework}
In the Maxwell framework with $v_0 = 0$, the variables corresponding to the Kerr–Newman solution take the values $\mathtt{x}_1 = \mathtt{x}_{KN}$, $\mathtt{x}_2 = \mathtt{y}_1 = \mathtt{z}_1 = \mathtt{z}_2 = 0$, and $\mathtt{y}_2 = \mathtt{y}_{KN}$, where:
\begin{subequations}\label{eq:Solution-KN-Maxwell}
\begin{equation}\label{eq:x_KerrNewman}
\mathtt{x}_{KN}
=-
\frac{1+\sigma l_1\overline{\xi}}{(\xi+l_1)\,(1-\sigma \xi \, \overline{\xi})}.
\end{equation}
\begin{equation}\label{eq:y_KerrNewman}
\mathtt{y}_{KN}
=
\frac{Q_0\,\sqrt{\sigma(\xi+l_1)}}{(\overline{\xi}+l_1)^{3/2}\sqrt{(\sigma\,\xi\bar\xi-1)}},
\end{equation}
\end{subequations}
here $l_1$ denotes the geometric mass, $Q_0=q+i p$, $q$ is the geometric electric charge, $p$ the geometric magnetic charge, and $\sigma=1/(l_1{}^2-|Q_0|^2)$ represent the curvature of the subspace $(\xi,\overline{\xi})$.
\subsection{Kerr Newman potentials in Frozen ModMax framework}
Assuming $v = v_0$ and $w = w_0$, and using the solution \eqref{eq:Solution-KN-Maxwell} of \eqref{N_EcuacionesDeCampo-abc}, we can reconstruct the potentials associated with the Kerr–Newman solution in the frozen ModMax framework by means of \eqref{N_RelacionesPotenciales-abc}:
{\small
\begin{equation}\label{eq:Potenciales-KN-MM}
    \begin{aligned}
        f_{KN}^{MM}&=f_0\,\frac{\sigma\xi\bar\xi-1}{(\xi+l_1)(\bar\xi+l_1)}=f_{KN}, 
        \\
        \psi_{KN}^{MM}&=\sqrt{\frac{f_0\sigma}{w_0}}\left(\frac{Q_0}{\xi+l_1}+\frac{\bar Q_0}{\bar\xi+l_1}\right)=\frac{\psi_{KN}}{\sqrt{w_0}},
        \\
        \chi_{KN}^{MM}&=i\frac{\sqrt{f_0\sigma}}{\sqrt{w_0}}\left[(w_0+iv_0)\frac{Q_0}{\xi+l_1}-(w_0-iv_0)\frac{\bar Q_0}{\bar\xi+l_1}\right] \nonumber \\
        &  =\sqrt{w_0}\chi_{KN}-\frac{v_0}{\sqrt{w_0}} \psi_{KN}
        \\
        \epsilon_{KN}^{MM}&=\frac{i f_0\sigma}{2}\left[\frac{Q_0^2-2l_1(\xi+l_1)}{(\xi+l_1)^2}-\frac{\bar Q_0^{\,2}-2l_1(\bar\xi+l_1)}{(\bar\xi+l_1)^2}\right] \\
        & \qquad \qquad -\frac{v_0 f_0\sigma}{2w_0} \left( \frac{Q_0}{\xi+l_1} + \frac{\bar Q_0}{\bar\xi+l_1} \right)^2 \nonumber \\
        &=\epsilon_{KN}-\frac{v_0}{2w_0}\psi_{KN}^2, 
    \end{aligned}
\end{equation}
}
where $f_0 = 1/\sigma$, which the subscript $KN$ denotes the Kerr–Newman configuration in the Maxwell framework.

However, from the definition \eqref{DefinicionPotencialesGeneral} it can be shown that $\omega_{KN}^{MM}=\omega_{KN}$ and $A_{\varphi-KN}^{MM}=\frac{1}{\sqrt{w_0}}A_{\varphi-KN}$. Consequently, this solution only introduces a rescaling in the electromagnetic potential of the Kerr–Newman solution.
%
\begin{theorem}[Trivialization of the frozen ModMax theory without scalar fields to a Maxwell-type theory]
Consider the frozen sector of ModMax theory characterized by \(v = v_0, w = w_0\), which corresponds to the case of constant \(\mathcal{F}/\mathcal{G}\). Let us further assume that the scalar field is turned off, i.e. \(\mathtt{z}_i = 0\), and that \(\alpha_0^2 = 0\), which implies that \(\kappa = \kappa_0\) is a constant.

Thus, the frozen electromagnetic sector is equivalent, via a constant-line redefinition of potentials, to the Maxwell-like sector.
More precisely, if we establish
{\small
\[ 
    \widehat{\psi}=\kappa_0\sqrt{w_0}\,\psi,
\qquad
\widehat{\chi}=\frac{\chi_{MM}+v_0\kappa_0^2\psi}{\kappa_0\sqrt{w_0}},
\qquad
\widehat{\epsilon}=\epsilon+\frac{v_0\kappa_0^2}{2}\psi^2.
\]
}
Therefore, the definition of the frozen potentials reduces precisely to the standard Maxwell-like terms, see \cite{Bixano:2026xum}.
\[
    (f,\widehat{\epsilon},\widehat{\psi},\widehat{\chi},\kappa_0 )
\]
in particular, the quadratures for $\omega, A_\varphi$ assume the standard Einstein–Maxwell form with $\kappa_0$ taken as a constant, and $A_\varphi$ is only re-normalized by a factor $\sqrt{w_0}$.
\end{theorem}
\begin{proof}
   Let us now consider the frozen ModMax theory in the absence of the scalar field. More precisely, we impose \(\mathtt{z}=0\), which is equivalent to fixing \(\kappa=\kappa_0\), while \(v=v_0\) and \(w=w_0\) are also taken to be constants. Under these assumptions, \eqref{TDchi} becomes
   \(\widetilde D\chi=w_0\,\widetilde D\chi_{MB}-v_0\kappa_0^{2}\,\widetilde D\psi\).
   Hence, up to an irrelevant additive gauge constant, we may write\(\chi=w_0\chi_{MB}-v_0\kappa_0^{2}\psi\) .
   We therefore define
   \[
      \widehat{\psi}=\kappa_0\sqrt{w_0}\,\psi,
      \qquad
      \widehat{\chi}=\frac{\chi+v_0\kappa_0^2\psi}{\kappa_0\sqrt{w_0}}.
   \]
   With this definition, we immediately obtain
   {\small
   \[
      \widehat{\chi}
      =
      \frac{w_0\chi_{MB}}{\kappa_0\sqrt{w_0}}
      =
      \frac{\sqrt{w_0}}{\kappa_0}\chi_{MB}, \quad \text{Implies} \quad  \widetilde D\widehat{\chi}
      =
      \frac{\sqrt{w_0}}{\kappa_0}\widetilde D\chi_{MB}.
   \]
   }
   Now, substituting \(\chi=\kappa_0\sqrt{w_0}\,\widehat{\chi}-v_0\kappa_0^2\psi\) into \eqref{TDepsilon}, we find
   \[
        \widetilde D\epsilon
      =
      \frac{f^2}{\rho}D\omega
      +
      \widehat{\psi}\,\widetilde D\widehat{\chi}
      -
      v_0\kappa_0^2\psi\,\widetilde D\psi .
   \]
   Since
   \(
      v_0\kappa_0^2\psi\,\widetilde D\psi
      =
      \widetilde D\!\left(\frac{v_0\kappa_0^2}{2}\psi^2\right),
   \)
   it follows that
   \[
      \frac{f^2}{\rho}D\omega
      =
      \widetilde D\!\left(\epsilon+\frac{v_0\kappa_0^2}{2}\psi^2\right)
      -
      \widehat{\psi}\,\widetilde D\widehat{\chi}.
   \]
   Thus, defining
   \[
      \widehat{\epsilon}
      :=
      \epsilon+\frac{v_0\kappa_0^2}{2}\psi^2,
   \]
   we obtain the standard Maxwell-like relation
   \[
      \frac{f^2}{\rho}D\omega
      =
      \widetilde D\widehat{\epsilon}
      -
      \widehat{\psi}\,\widetilde D\widehat{\chi}.
   \]

   Therefore, in the frozen sector with \(z=0\), the ModMax deformation is absorbed into the linear redefinition
   \[
      (\psi,\chi,\epsilon)\mapsto
      (\widehat{\psi},\widehat{\chi},\widehat{\epsilon}),
   \]
   and the resulting system takes exactly the Maxwell-like form.
      Finally, since \(\widehat{\psi}=\kappa_0\sqrt{w_0}\,\psi,\)
      and \(A_t=\frac{\psi}{2}\), it follows that \( \widehat{A}_t=\frac{\widehat{\psi}}{2}= \kappa_0\sqrt{w_0}\,A_t \), that is,
   \[
      A_t=\frac{1}{\kappa_0\sqrt{w_0}}\widehat{A}_t.
   \]
   On the other hand, from the definition
   \(
      \widetilde D\chi_{MB}
      =
      \frac{2f\kappa_0^2}{\rho}
      \bigl(\omega\,DA_t+DA_\varphi\bigr),
   \)
   and from
   \(
      \widehat{\chi}
      =
      \frac{\sqrt{w_0}}{\kappa_0}\chi_{MB},
   \)
   we obtain
   \[
      \widetilde D\widehat{\chi}
      =
      \frac{2f}{\rho}
      \Bigl(
      \omega\,D(\kappa_0\sqrt{w_0}\,A_t)
      +
      D(\kappa_0\sqrt{w_0}\,A_\varphi)
      \Bigr).
   \]
   Therefore, defining \(\widehat{A}_\varphi=\kappa_0\sqrt{w_0}\,A_\varphi\), 
   the quadrature takes the standard Maxwell-like form defining 
   \[
      A_\varphi=\frac{1}{\kappa_0\sqrt{w_0}}\widehat{A}_\varphi.
   \]

   In conclusion, both electromagnetic components of the four-potential are normalized by the same constant factor,
   \[
      A_\mu=w_0^{-1/2}\widehat{A}_\mu.
   \]
\end{proof}
\section{First class of solutions}\label{Sec:First class of solutions}
The first class of solutions arises by assuming that $\mathtt{x},\mathtt{y},\mathtt{z}$ are complex numbers, and  imposing $\sigma=0$. In this case, solving the system \eqref{N_EcuacionesDeCampo-abc} algebraically yields the following solution:
\begin{equation}\label{eq:Solution-First-Class}
    \begin{aligned}
        &\mathtt{x}=\mathtt{x}_1=\mathtt{x}_2=i q, \qquad \mathtt{z}=\mathtt{z}_1=\mathtt{z}_2=\frac{2q}{\eta_0 \pm \sqrt{1+\eta_0^2}}, \\
        &\mathtt{y}=\mathtt{y}_1=\mathtt{y_2}=\sqrt{2} \, e^{i \, \Xi} q, \qquad 2\Xi=-\arctan{\eta_0}\, \pm \frac{\pi }{2} ,
    \end{aligned}
\end{equation}
where $q\in \mathbb{R}$.

This solution matches the fourth class of solutions discussed in \cite{Matos:2000za} and the second class of solutions discussed in \cite{Matos:2010pcd}, now generalized to the ModMax framework. By reconstructing the potentials, we arrive at the following expressions:
\begin{equation}\label{eq:Potentials-First-Class}
    \begin{aligned}
        &f=f_0, \qquad \kappa=\kappa_0 e^{-\mathtt{z} \, \mathfrak s}, \\
        &\psi=\psi_0 -\frac{2\sqrt{f_0}}{\kappa_0 \mathtt{z} \sqrt{w_0}} \;  \Re \mathtt{y}\; e^{\mathtt{z} \, \mathfrak s},\\
        &\chi=\chi_0 -\frac{2\kappa_0 \sqrt{f_0 w_0}}{\mathtt{z}}(\Im \mathtt{y}+\eta_0 \; \Re \mathtt{y})e^{-\mathtt{z} \, \mathfrak s},\\
        &\epsilon=\epsilon_1-\frac{2\psi_0 \kappa_0 \sqrt{f_0 w_0}}{\mathtt{z}} (\Im \mathtt{y} + \eta_0 \, \Re \mathtt{y}) e^{-\mathtt{z} \, \mathfrak s}.
    \end{aligned}
\end{equation}
where $f_0,\epsilon_1,\psi_0,\chi_0,\kappa_0$ are real constants, $\mathfrak s=(\xi+\overline{\xi})/2$ is a harmonic function that represents a solution of the Laplace equation in flat space \eqref{Eq:LaplaceXi}, and, for the specific gauge choice $\psi_0=0$, the potential $\epsilon$ becomes constant.

This solution represents a \emph{Geon}, namely a compact object composed purely of an electromagnetic field together with a scalar field, but lacking any gravitational field. An interesting feature is that, when determining the metric function $\omega$, it turns out to be non-constant, even though the rotational potential $\epsilon=\epsilon_1$ remains constant.

If we set $v_0=0$, then $\eta_0$ also vanishes, and we reproduce the solution family presented in \cite{Matos:2010pcd,Matos:2000za} and further analyzed in \cite{Bixano:2025jwm,Bixano:2025bio,Bixano:2025qxp,DelAguila:2018gni,DelAguila:2015isj}.
\section{Second class of solutions}\label{Sec:Second class of solutions}
The second type of solutions consists of using a real harmonic function $\mathfrak s$ that satisfies the Laplace equation in flat spaces, that is, with $\sigma = 0$. When imposing the conditions $\mathtt{x}_1=\mathtt{x}_2=\mathtt{x}(\mathfrak{s})\in \mathbb{R}$, $\mathtt{y}_1=\mathtt{y}_2=\mathtt{y}(\mathfrak{s})\in \mathbb{R}$ and $\mathtt{z}_1=\mathtt{z}_2=\mathtt{z}(\mathfrak{s})\in \mathbb{R}$,the following solution is obtained:
{\small
\begin{equation}\label{eq:Solution-Second-Class}
    \begin{aligned}
        &\mathtt{x}=-\frac{1}{(1+a_0)\mathfrak{s}}, \quad \mathtt{y}=\pm \frac{1}{\sqrt{1+a_0}}\frac{1}{\mathfrak{s}}, \quad
        &\mathtt{z}=-\frac{a_0}{(1+a_0)\mathfrak{s}},
    \end{aligned}
\end{equation}
}
where $a_0=\alpha_0^2/\epsilon_0$.

Using \eqref{N_RelacionesPotenciales-abc}, we obtain the following potentials:
\begin{equation}\label{eq:Potentials-Second-Class}
    \begin{aligned}
        &f=f_0 \mathfrak{s}^{-2/(1+a_0)}, \qquad \kappa=\kappa_0 \mathfrak{s}^{a_0/(1+a_0)}, \\
        &\psi=\psi_0 \pm \frac{2\sqrt{f_0}}{\kappa_0  \sqrt{w_0} \, \sqrt{1+a_0 }} \;  \mathfrak s^{-1},\\
        &\chi=\chi_0 \pm \frac{2\eta_0\kappa_0\sqrt{f_0w_0}\sqrt{1+a_0}}{a_0-1}\,\mathfrak{s}^{(a_0-1)/(a_0+1)},\\
        &\epsilon=\epsilon_1\pm \psi_0\,\frac{2\eta_0\kappa_0\sqrt{f_0w_0}\sqrt{1+a_0}}{a_0-1} \mathfrak{s}^{(a_0-1)/(a_0+1)} \\
        & \qquad \qquad -2\eta_0 f_0\,\mathfrak{s}^{-2/(a_0+1)}.
    \end{aligned}
\end{equation}
The remarkable aspect of this type of solution is that, with this new family, we can clearly observe a very interesting feature: it makes explicit the distinction between frozen ModMax and Maxwell. In fact, if we set $\eta_0 = 0$, the magnetic potential $\chi$ and the rotational potential $\epsilon$ become constant, which in turn implies, via \eqref{DefinicionPotencialesGeneral}, that the solution reduces to a static, purely electric configuration. In other words, within Maxwell theory we have $\omega = 0$ and $A_\varphi = 0$. However, when $\eta_0 \neq 0$, i.e. in the frozen ModMax case, the solution indeed possesses a nontrivial magnetic potential and rotation.
\section{Third class of solutions}\label{Sec:Third class of solutions}
In this case, we derive a new solution analogous to the second class of solutions ($\sigma=0$), but now we impose a phase on the variables $\mathtt{y}$. Specifically, we set $\mathtt{x}_1=\mathtt{x}_2=\mathtt{x}(\mathfrak{s})\in \mathbb{R}$, $\mathtt{y}_1=\mathtt{y}_2=\mathtt{r}(\mathfrak{s}) e^{i \, \Xi}$, and $\mathtt{z}_1=\mathtt{z}_2=\mathtt{z}(\mathfrak{s})\in \mathbb{R}$, which yields the following solution:
{\small
\begin{equation}\label{eq:Solution-Second-Class}
    \begin{aligned}
        &\mathtt{x}=-\frac{1}{(1+a_0)\mathfrak{s}}, \quad \mathtt{r}=\pm \frac{1}{\sqrt{1+a_0}}\frac{1}{\mathfrak{s}}, \quad
        &\mathtt{z}=\frac{a_0}{(1+a_0)\mathfrak{s}}, \\
        & \Xi = \arctan{\frac{1}{\eta_0}}, \qquad  S_\eta=\sin{\Xi}+\eta_0 \cos{\Xi},
    \end{aligned}
\end{equation}
}
The corresponding potentials will be:
{\small
\begin{equation}\label{eq:Potentials-Second-Class}
    \begin{aligned}
        &f=f_0 \mathfrak{s}^{-2/(1+a_0)}, \qquad \kappa=\kappa_0 \mathfrak{s}^{-a_0/(1+a_0)}, \\
        &\psi=\psi_0 \mp \frac{2 \sqrt{f_0}\sqrt{1+a_0} \; \cos{\Xi}}{(a_0-1)\kappa_0 \sqrt{{w_0}}}\,\mathfrak{s}^{(a_0-1)/(a_0+1)} ,\\
        &\chi=\chi_0 \mp \frac{2 \kappa_0 \sqrt{f_0 w_0} \; S_\eta}{ \sqrt{1+a_0 }} \;   \mathfrak s^{-1},\\
        &\epsilon=\epsilon_1 \mp \psi_0 \frac{2\kappa_0\sqrt{f_0w_0}}{\sqrt{1+a_0}} S_\eta\,\mathfrak{s}^{-1}+\frac{2(1+a_0)f_0}{a_0-1}\cos{\Xi}\,S_\eta\,\mathfrak{s}^{-2/(a_0+1)}.
    \end{aligned}
\end{equation}
}
In this situation, we again observe the explicit phenomenology associated with the frozen ModMax theory. Specifically, by choosing $\eta_0 = 0$, we obtain $\Xi = 0$ and $S_\eta = 0$, which in turn implies that the potentials $\chi,\epsilon$ must be constant.

\section{Conclusions}   %
In this work, we employ a Newman–Penrose-type coframe on the space of potentials extended to ModMax, through which we derive the 1-forms $ABC$ to cast the field equations into a more compact form and thereby investigate their structure, following the approach of \cite{Bixano:2026xum}. Using these methods, we obtain the generalization of the Einstein–ModMax–scalar field system applicable to any theory.

Using the equations projected onto the ($\xi,\overline{\xi}$) flat subspace in the frozen ModMax case, we were able to construct three families of new solutions that feature rotation ($\omega \neq 0$), a nontrivial electromagnetic field, and a coupled scalar field. Families 2 and 3 are dual to each other and closely related, yet they do not correspond to the same class of solutions.

By analyzing equations \eqref{N_EcuacionesDeCampo-abc}, we found that in the frozen ModMax theory it is necessary to switch on the scalar field in order to exhibit the characteristic ModMax phenomenology, and thereby avoid reducing to the trivial Maxwell-like scenario.

Similarly, the symmetries of this theory coincide precisely with those described in Ref. \cite{Bixano:2026xum}. It was shown that, when restricting to the frozen ModMax theory, it admits a chiral matrix formulation only in the case $v_0 = 0$; for $v \neq 0$, such a reformulation is not possible.
\section{Acknowledgements}

LB thanks SECIHTI-M\'exico for the doctoral grant.
This work was also partially supported by SECIHTI M\'exico under grants SECIHTI CBF-2025-G-1720 and CBF-2025-G-176. 


\begin{appendices}
\section{Coordinates protection}\label{Apendix:Coordinates protection}
The Weyl ansatz metric \eqref{ds Cilindricas}, when written in oblate ($+$) or prolate ($-$) spheroidal coordinates $(x, y)$, is given by
\begin{multline}\label{ds sp}
    ds^2 = -f\left( dt-\omega d \varphi \right)^2
     + \frac{(L_{\pm})^2}{f} \bigg( (x^2\pm1)(1-y^2) d\varphi^2 \\
     +(x^2\pm y^2) e^{2k} \left\{ \frac{dx^2}{x^2\pm 1} +\frac{dy^2}{1-y^2} \right\} \bigg).
\end{multline}
the Weyl coordinates are connected to $(x, y)$ by the following relation
\begin{equation}\label{WeylSpheroidalCordiantes}
\rho=(L_{\pm})\sqrt{(x^2\pm1)(1-y^2)} , \qquad z=(L_{\pm})xy, 
\end{equation}
and $\rho \in [0,\infty)$, $\{ z ,x \}\in \mathbb{R}$, $y \in [-1,1]$, $(L_{\pm}) \geq 0$.

There are two scenarios

\paragraph{Sub-extreme (S-E: lower sign $-$) condition}:
\begin{align}\label{Sub-Extreme}
|\mathcal{M}_\infty|^{2}&>a^2+Q_{L}^{2}+H_{L}^{2} ,\\
|\mathcal{M}_\infty|^{2}&=L_{-}^2+a^2+Q_{L}^{2}+H_{L}^{2},
\end{align}

\paragraph{Super-extreme (SU-E: upper sign $+$) condition}:
\begin{align}\label{Super-Extreme}
|\mathcal{M}_\infty|^{2}&<a^2+Q_{L}^{2}+H_{L}^{2} ,\\
L_{+}^2+|\mathcal{M}_\infty|^{2}&=a^2+Q_{L}^{2}+H_{L}^{2}.
\end{align}
Here \(a = J_\infty / |\mathcal{M}_\infty|\) represents the angular momentum per unit effective mass, where \(\mathcal{M}_\infty = l_1 + i\,N_\infty\), with \(l_1\) being a length-scale parameter, and \(Q_L = Q_\infty\), \(H_L = H_\infty\) denoting the electric and magnetic geometric charges, respectively, while \(N_\infty\) is the NUT parameter. The quantities \(Q_\infty, H_\infty, N_\infty,\) and \(J_\infty\) are conserved invariant charges, defined in the usual classical sense as in \cite{Komar:1958wp,Nedkova:2011hx,Clement:2015aka,Clement:2022pjr,BallonBordo:2019vrn,Misner:1963fr,Manko:2005nm}. In typical situations, the Komar mass \(M_\infty\) is equal to \(l_1\).
Finally, the Boyer–Lindquist coordinates $(r,\theta)$ can be expressed in terms of the preceding coordinate system as
\begin{equation}\label{BoyerLindquistSpheroidalCordiantes}
    (L_{\pm}) x=r-l_1, \qquad y=\cos{\theta},
\end{equation}
where $r \in (-\infty,-l_1]\cup[l_1,\infty)$, $\theta\in [0,\pi]$, and the variable $l_1=r_s/2$, where $r_s$ represents the Schwarzschild radius.
\subsection{Potential differential equations}
Using the relation \eqref{WeylSpheroidalCordiantes}, the potential equations \eqref{DefinicionPotencialesGeneral}, when written in spheroidal coordinates, assume the following form:
{\small
\begin{subequations}\label{FuncionesMetricas-xy}
\begin{center}
    \text{\textbf{Differential equation for $\omega$}}
    \begin{equation}\label{EcuacionesDiferencialesOmega}
        \begin{bmatrix}
            \partial_x \\
            \partial_y \\
        \end{bmatrix} 
        \omega= \frac{1}{f^2}\Big( \begin{bmatrix}
            (1-y^2)\partial_y  \\
            -(x^2\pm1)\partial_x  
        \end{bmatrix}\epsilon -\psi \begin{bmatrix}
            (1-y^2)\partial_y  \\
            -(x^2\pm1)\partial_x  
        \end{bmatrix}\chi \Big) ,
    \end{equation}

    \text{\textbf{Differential equation for $A_\varphi$}}
    \begin{equation}\label{EcuacionesDiferencialesA3}
        \begin{bmatrix}
            \partial_x \\
            \partial_y \\
        \end{bmatrix} 
        A_\varphi = \frac{1}{2  f \kappa^2 } \begin{bmatrix}
            (1-y^2)\partial_y  \\
            -(x^2\pm1)\partial_x  
        \end{bmatrix} \chi
        -\frac{\omega}{2} 
        \begin{bmatrix}
            \partial_x \\
            \partial_y \\
        \end{bmatrix} \psi.
    \end{equation}
    \end{center}
\end{subequations}
}
The metric function $k$ can be expressed by the following integral
{\footnotesize
\begin{subequations}\label{kVariableCompleja}
\begin{align}
    k_{,\zeta}=& \frac{\rho}{4f^2} \bigg[ \tensor{f}{_{,\zeta}}{^2}+(\tensor{\epsilon}{_{,\zeta}} -\psi \tensor{\chi}{_{,\zeta}})^2 \bigg] +\frac{2 \rho \epsilon_0}{\alpha_0 ^2 \, \kappa^2 }\tensor{\kappa}{_{,\zeta}}{^2}
    \\
    &-\frac{\rho}{2 f}\bigg[ \frac{w^2+v^2}{w}\kappa^2 \tensor{\psi}{_{,\zeta}}{^2} +\frac{\tensor{\chi}{_{,\zeta}}{^2}}{w \, \kappa^2} +2 \frac{v}{w} \tensor{\psi}{_{,\zeta}} \tensor{\chi}{_{,\zeta}}\bigg],\label{k Chi}
\end{align}
\begin{align}
    k_{,\overline \zeta}=& \frac{\rho}{4f^2} \bigg[ \tensor{f}{_{,\overline \zeta}}{^2}+(\tensor{\epsilon}{_{,\overline \zeta}} -\psi \tensor{\chi}{_{,\overline \zeta}})^2 \bigg] +\frac{2 \rho \epsilon_0}{\alpha_0 ^2 \, \kappa^2 }\tensor{\kappa}{_{,\overline \zeta}}{^2}
    \\
    &-\frac{\rho}{2 f}\bigg[ \frac{w^2+v^2}{w}\kappa^2 \tensor{\psi}{_{,\overline \zeta}}{^2} +\frac{\tensor{\chi}{_{,\overline \zeta}}{^2}}{w \, \kappa^2} +2 \frac{v}{w} \tensor{\psi}{_{,\overline \zeta}} \tensor{\chi}{_{,\overline \zeta}}\bigg].\label{k complejaconjugada}
\end{align}
\end{subequations}
}
where $\zeta=\rho +i z$.

Note that each quadratic term arising from a partial derivative, $\mathtt{F}_{,\zeta}$ or $\mathtt{F}_{,\overline \zeta}$ with $\mathtt F \in \{f,\epsilon,\psi,\chi,\kappa \}$, and with coeficients
{\small
\[
    \mathtt G_{\mathtt{F}} \in \left\{
    \frac{1}{4f^2}, \;
    \frac{1}{4f^2}, \;
    -\frac{\kappa^2}{2f}\frac{v^2+w^2}{w}, \;
    \frac{\psi^2}{4f^2}-\frac{1}{2fw\kappa^2}, \;
    \frac{2\epsilon_0}{\alpha_0^2 \kappa^2}
     \right\},
\]
}
successively takes on the following form when expressed in spheroidal coordinates:
{\small
\begin{subequations}\label{kCadaTerminoCuadradito}
\begin{align}
    &\partial_x k \propto \mathtt G_{\mathtt{F}} \; \frac{(1-y^2)}{x^2\pm y^2} \bigg\{ x\Big[ (x^2\pm1)(\partial_x \mathtt F)^2-(1-y^2) (\partial_y \mathtt F)^2  \Big] \nonumber \\
    & \qquad -2y (x^2\pm1)(\partial_x \mathtt F)(\partial_y \mathtt F) \bigg\}, \label{EcuacionesDiferencialesk1Completo}\\
    &\partial_y k \propto \mathtt G_{\mathtt{F}} \;  \frac{(x^2\pm 1)}{x^2\pm y^2} \bigg\{ y \Big[ (x^2\pm 1)(\partial_x \mathtt F)^2-(1-y^2) (\partial_y \mathtt F)^2  \Big] \nonumber \\ 
    & \qquad + 2x (1-y^2)(\partial_x \mathtt F)(\partial_y \mathtt F) \bigg\} \label{EcuacionesDiferencialesk2Completo} ,
\end{align}
\end{subequations}
}
Similarly, every mixed bilinear contribution of the form
$\mathtt F_{,\zeta}\mathtt H_{,\zeta}$, with coefficient
$\mathtt G_{\mathtt F\mathtt H}$, adopts the universal spheroidal structure
{\footnotesize
\begin{subequations}\label{kCadaTerminoCruzadito}
\begin{align}
    &\partial_x k \propto \mathtt G_{\mathtt F\mathtt H}\,
    \frac{(1-y^2)}{x^2\pm y^2}
    \bigg\{
    x\Big[
    (x^2\pm1)(\partial_x\mathtt F)(\partial_x\mathtt H)
    -(1-y^2)(\partial_y\mathtt F)(\partial_y\mathtt H)
    \Big]
    \nonumber\\
    &\qquad\qquad
    -y(x^2\pm1)\Big[
    (\partial_x\mathtt F)(\partial_y\mathtt H)
    +(\partial_y\mathtt F)(\partial_x\mathtt H)
    \Big]
    \bigg\},
    \label{EcuacionesDiferencialesk1CruzadoCompleto}
    \\[0.3em]
    &\partial_y k \propto \mathtt G_{\mathtt F\mathtt H}\,
    \frac{(x^2\pm1)}{x^2\pm y^2}
    \bigg\{
    y\Big[
    (x^2\pm1)(\partial_x\mathtt F)(\partial_x\mathtt H)
    -(1-y^2)(\partial_y\mathtt F)(\partial_y\mathtt H)
    \Big]
    \nonumber\\
    &\qquad\qquad
    +x(1-y^2)\Big[
    (\partial_x\mathtt F)(\partial_y\mathtt H)
    +(\partial_y\mathtt F)(\partial_x\mathtt H)
    \Big]
    \bigg\},
    \label{EcuacionesDiferencialesk2CruzadoCompleto}
    .
\end{align}
\end{subequations}
}
Here $\mathtt G_{\mathtt F\mathtt H}$ denotes the coefficient multiplying the mixed term
$\mathtt F_{,\zeta}\mathtt H_{,\zeta}$ in the Weyl quadrature. In the present case, the non-vanishing mixed contributions are
\[
(\mathtt F,\mathtt H,\mathtt G_{\mathtt F\mathtt H})
=
\left(\epsilon,\chi,-\frac{\psi}{2f^2}\right),
\qquad
\left(\psi,\chi,-\frac{v}{fw}\right).
\]
The corresponding terms involving
$\mathtt F_{,\overline\zeta}\mathtt H_{,\overline\zeta}$ are obtained analogously, with the same coefficient
$\mathtt G_{\mathtt F\mathtt H}$.

Therefore, the full $k$-quadrature is obtained by collecting the purely quadratic terms from \eqref{kCadaTerminoCuadradito} and then adding the mixed contribution from \eqref{kCadaTerminoCruzadito}.
\section{Structure of the space of potentials}
\label{Appendix:Structure of the space of potentials}
Before presenting the explicit formulas, we first introduce the notation that will be used:
\begin{equation}\label{eq:Def_Lambda_Nu}
    \lambda=\mathrm d(\ln w)=\frac{\mathrm d w}{w}, 
\qquad
\nu=\frac{\mathrm d v}{w},  \qquad \eta=\frac{v}{w}.
\end{equation}
Consequently, we obtain
\[
    \dd \lambda=0, \qquad \dd \nu=-\lambda \wedge \nu, \qquad \dd \eta =\nu-\eta \lambda. 
\]
\subsection{Connection 1-forms}
From the relation
\(
\dd \theta^I + \omega^I{}_{J} \wedge \theta^J = 0,
\)
we have derived the connection 1-forms associated with the Newman–Penrose-like null coframe \eqref{eq:Coframe-NP-like}, which are given by
{\footnotesize
\begin{equation}\label{eq:1Forma-Conexion}
    \begin{aligned}
        &\omega^n{}_{\ell}=-n,\qquad\omega^\ell{}_{n}=-\ell, \qquad \omega^n{}_{m}=i\,\bar m,\qquad \omega^\ell{}_{\bar m}=i\,m,
        \\
        &\omega^{\bar m}{}_{\ell}=-\frac12\,\bar m,
        \qquad 
        \omega^{\bar m}{}_{n}=-\frac12\,\bar m, 
        \\
        &\omega^m{}_{\ell}=-\frac12\,m,
        \qquad 
        \omega^m{}_{n}=-\frac12\,m, 
        \\
        &\omega^m{}_{m}=-\frac{i}{2}\nu-\frac{i\alpha_0}{\sqrt2}\eta\,s, 
        \qquad 
        \omega^{\bar m}{}_{\bar m}=+\frac{i}{2}\nu+\frac{i\alpha_0}{\sqrt2}\eta\,s,
        \\
        &\omega^m{}_{\bar m}=-\frac12(\lambda+i\nu),
        \qquad 
        \omega^{\bar m}{}_{m}=-\frac12(\lambda-i\nu), 
        \\
        &\omega^m{}_{s}=\frac{\alpha_0}{\sqrt2}(1+i\eta)\,\bar m,
        \qquad 
        \omega^{\bar m}{}_{s}=\frac{\alpha_0}{\sqrt2}(1-i\eta)\,m,
        \\
        &\omega^{s}{}_{m}=\frac{\alpha_0}{\sqrt2\,\epsilon_0}(1-i\eta)\,m, 
        \qquad 
        \omega^{s}{}_{\bar m}=\frac{\alpha_0}{\sqrt2\,\epsilon_0}(1+i\eta)\,\bar m.
    \end{aligned}
\end{equation}
}
All the other vanishes.
\subsection{Curvature 2-forms}
Using the relation
\(
\Omega^I{}_{J} = \mathrm d\omega^I{}_{J} + \omega^I{}_{K} \wedge \omega^K{}_{J}
\), we then derived the associated curvature 2-forms:
{\footnotesize
\begin{equation}\label{eq:2Forma-Curvatura}
    \begin{aligned}
    &\Omega^n{}_{\ell}
    =
    \ell\wedge n-\frac{i}{2}\,m\wedge\bar m,
    \qquad
    \Omega^\ell{}_{n}
    =
    -\ell\wedge n+\frac{i}{2}\,m\wedge\bar m,
    \\
    &\Omega^{\bar m}{}_{\ell}
    =
    \frac14\,\ell\wedge\bar m
    -\frac14\,n\wedge\bar m
    +\frac{\alpha_0}{2\sqrt2}(1-i\eta)\,m\wedge s,
    \\
    &\Omega^{\bar m}{}_{n}
    =
    -\frac14\,\ell\wedge\bar m
    +\frac14\,n\wedge\bar m
    +\frac{\alpha_0}{2\sqrt2}(1-i\eta)\,m\wedge s,
    \\
    &\Omega^{m}{}_{\ell}
    =
    \frac14\,\ell\wedge m
    -\frac14\,n\wedge m
    +\frac{\alpha_0}{2\sqrt2}(1+i\eta)\,\bar m\wedge s,
    \\
    &\Omega^{m}{}_{n}
    =
    -\frac14\,\ell\wedge m
    +\frac14\,n\wedge m
    +\frac{\alpha_0}{2\sqrt2}(1+i\eta)\,\bar m\wedge s, 
    \\
    &\Omega^{m}{}_{m}
    =
    -\frac{i\alpha_0}{\sqrt2}\left(\nu-\eta\lambda\right)\wedge s
    -
    \left[
        \frac{i}{2}
        +
        \frac{\alpha_0^2}{2\epsilon_0}
        \left(
            1+\eta^2
        \right)
    \right]
    m\wedge\bar m,
    \\
    &\Omega^{\bar m}{}_{\bar m}
    =
    +\frac{i\alpha_0}{\sqrt2}\left(\nu-\eta\lambda\right)\wedge s
    +
    \left[
        \frac{i}{2}
        +
        \frac{\alpha_0^2}{2\epsilon_0}
        \left(
            1+\eta^2
        \right)
    \right]
    m\wedge\bar m,
    \\
    &\Omega^{m}{}_{\bar m}
    =
    \frac{i\alpha_0}{\sqrt2}\eta\,s\wedge\lambda
    -\frac{\alpha_0}{\sqrt2}\eta\,s\wedge\nu,
    \\
    &\Omega^{\bar m}{}_{m}
    =
    -\frac{i\alpha_0}{\sqrt2}\eta\,s\wedge\lambda
    -\frac{\alpha_0}{\sqrt2}\eta\,s\wedge\nu,
    \\
    &\Omega^{\bar m}{}_{s}
    =
    -\frac{\alpha_0}{2\sqrt2}
    \left(
        1-i\eta
    \right)
    (\ell+n)\wedge m
    \\
    &\qquad\qquad
    +
    \left[
        \frac{\alpha_0\eta}{\sqrt2}
        \left(
            \nu+i\lambda
        \right)
        +
        \alpha_0^2\eta(\eta+i)\,s
    \right]\wedge m
    \\
    &\qquad\qquad
    +
    \left[
        \frac{i\alpha_0}{\sqrt2}
        \left(
            \nu-\eta\lambda
        \right)
        +
        \frac{\alpha_0^2}{2}
        \left(
            1+\eta^2
        \right)s
    \right]\wedge\bar m,
    \\
    &\Omega^{m}{}_{s}
    =
    -\frac{\alpha_0}{2\sqrt2}
    \left(
        1+i\eta
    \right)
    (\ell+n)\wedge \bar m
    \\
    &\qquad\qquad
    +
    \left[
        \frac{\alpha_0\eta}{\sqrt2}
        \left(
            \nu-i\lambda
        \right)
        +
        \alpha_0^2\eta(\eta-i)\,s
    \right]\wedge \bar m
    \\
    &\qquad\qquad
    +
    \left[
        -\frac{i\alpha_0}{\sqrt2}
        \left(
            \nu-\eta\lambda
        \right)
        +
        \frac{\alpha_0^2}{2}
        \left(
            1+\eta^2
        \right)s
    \right]\wedge m.
    \end{aligned}
\end{equation}
}
All the other vanishes.
\end{appendices}
\bibliographystyle{elsarticle-harv} 
\bibliography{Bibliografia}

\end{document}